\documentclass[11pt]{article}
\usepackage[T2A]{fontenc}
\usepackage[cp1251]{inputenc}
\usepackage[english]{babel}
\usepackage{amssymb}
\usepackage{amsmath}
\usepackage{amsthm}
\usepackage{amsfonts}
\usepackage{amssymb}

\newtheorem{theorem}{Theorem}

\newtheorem{proposition}{Proposition}

\newtheorem{definition}{Definition}

\newtheorem{remark}{Remark}

\DeclareMathOperator*{\indlim}{ind \lim}

\begin{document}
\begin{center}
\textbf{\uppercase{L\'{e}vy Laplacians  in Hida calculus and Malliavin calculus}}

B.~O.~Volkov

\texttt{borisvolkov1986@gmail.com}

Steklov Mathematical Institute of Russian Academy of Sciences,\\ ul. Gubkina 8, Moscow, 119991 Russia

\end{center}

{\bf Abstract:} 
Some connections between different definitions of L\'{e}vy Laplacians in the stochastic analysis are considered. Two approaches are used to define these operators. The standard one is based on the application of the theory of Sobolev-Schwartz distributions over the Wiener measure (the Hida calculus).  One can consider the  chain of L\'evy Laplacians parametrized by a real parameter with the help of this approach.  One of the elements of this chain is the classical L\'evy Laplacian.  Another approach to define the L\'evy Laplacian is based on the application of the theory of Sobolev spaces over the Wiener measure (the Malliavin calculus). It is proved that the L\'evy Laplacian defined with the help of the second approach coincides with one of the elements of the chain of L\'evy Laplacians, which is not the classical L\'evy Laplacian, under the imbedding of the Sobolev space over the Wiener measure into the space of generalized functionals over this measure. It is shown  which L\'{e}vy Laplacian in the stochastic analysis is connected to the gauge fields.

key words: L\'{e}vy Laplacian, Yang-Mills equations, Hida calculus, Malliavin calculus 


\section*{Introduction}
An infinite-dimensional Laplacian, given by the formula
\begin{equation}
\label{l1}
\Delta^{\{e_n\}}_L f(x)=\lim_{n\to \infty}\frac 1n\sum_{k=1}^n \langle f''(x)e_k,e_k\rangle,
\end{equation}
where the function $f$  is defined on a separable Hilbert space  $H$ and  $\{e_n\}$ is an orthonormal basis in  $H$, is called the  L\'{e}vy Laplacian or the L\'{e}vy-Laplace operator. 
The Hida calculus or the white noise analysis is the theory of Sobolev-Schwartz distributions over the (abstract) Wiener measure. The Malliavin calculus is the theory of Sobolev spaces over measures on infinite-dimensional spaces, in particular, over the Wiener measure.\footnote{In addition, the Malliavin calculus includes an analysis of the smoothness of non-linear images of measures that is not used in this paper.} The goal of the present paper is to show that  there are two infinite-dimensional Laplacians in the stochastic analysis, which are natural analogues of  operator~(\ref{l1}). The first Laplacian is defined
using  the Hida calculus. This operator will be called the classical L\'evy Laplacian. An extensive literature is devoted to the study of this operator (see the review~\cite{K2003}, and also~\cite{Volkov2013,AJS2013} and the papers cited there). The second Laplacian is defined using the Malliavin calculus. This operator is connected with the gauge fields. 
In the Hida calculus one can consider a chain of infinite-dimensional Laplacians parametrized by a real parameter. One of the elements of this chain is the classical L\'evy Laplacian. In the paper it is proved that the L\'evy Laplacian, defined with the help of the Malliavin calculus, is isomorphic to the element of this chain under the embedding of the Sobolev space over the Wiener measure
into the space of generalized functionals over this measure. This element does not coincide with the classical L\'evy Laplacian.

One of the main reasons for interest in the L\'{e}vy Laplacian and differential operators defined by
analogy is their connection to the Yang-Mills equations. In the papers~\cite{AGV1993, AGV1994} by Accardi, Gibilisco and Volovich an analogue of the L\'{e}vy Laplacian on the space of functions on the space of paths in $\mathbb{R}^d$ has been introduced. This analogue is also called the L\'{e}vy Laplacian. For this operator in~\cite{AGV1994} the following has been proved. A connection in a vector bundle over $\mathbb{R}^d$ is a solution to the Yang-Mills equations if and only if the parallel transport generated by this connection is a solution to the Laplace equation for the L\'{e}vy Laplacian. In the paper~\cite{LV2001} by  Leandre and  Volovich the L\'{e}vy Laplacians on the space of functions on the set of paths in a compact Riemannian manifold and on the Sobolev space over the Wiener measure  on the space of paths in a compact Riemannian manifold have been introduced. It has been shown  that the theorem on the connection between the Levy Laplacian and the gauge fields is also satisfied in these cases.
 In~\cite{VolkovLLI}  by the author  the relationship between the L\'{e}vy Laplacian   and instantons has been studied.

In~\cite{AGV1993,AGV1994,LV2001} the L\'{e}vy Laplacian has been defined not as the  Ces\`aro mean of the second derivatives but as an integral functional generated by a special form of the second derivative. This approach to  define the L\'{e}vy Laplacian also goes back to the original P. L\'{e}vy's works. In  the paper~\cite{AGV1993} the following problem has been posed:
can one represent the L\'{e}vy Laplacian associated with the gauge fields in a form similar to~(\ref {l1}). In the deterministic planar case, it has been shown in work~\cite{VolkovLLI} that this is indeed the case (see also~\cite{AS2006}).
In the author's work~\cite{Volkov2017}, using the Malliavin calculus, the L\'{e}vy Laplacian, defined as the Ces\`aro mean of second partial derivatives, has been introduced on the Sobolev space over the Wiener measure and its relation to gauge fields has been studied. It should be noted that, unlike the deterministic case, the L\'{e}vy Laplacians on the Sobolev spaces over the Wiener measure, introduced in the works~\cite{LV2001} and~\cite{Volkov2017}, operate in different ways.

One of the motivations for studying the Hida calculus was the development of the harmonic analysis for the L\'{e}vy Laplacian (see~\cite{H1975}, and also~\cite{KOS,HKPS,K,HiSi}). If a Sobolev space over a Wiener measure is a natural domain of the Gross-Volterra Laplacian, then the space of generalized Hida functionals is a sufficiently wide space to contain the  domain of the classical L\'{e}vy Laplacian. In addition, using the Hida calculus, certain generalizations of the L\'{e}vy Laplacian has been studied: the so-called exotic and nonclassical L\'{e}vy Laplacians (see~\cite{AJS2011,AJS2013,Volkov2013}). The family of the exotic L\'{e}vy Laplacians $\Delta^{(s)}_{exotic}$, where $s\geq 0$, has been  introduced in~\cite{AS1993}. This family has the following properties: 
$\Delta^{(0)}_{exotic}$ is the Gross-Volterra Laplacian; $\Delta^{(1)}_{exotic}$ is the classical L\'{e}vy Laplacian; if $s_1<s_2$, then  the domain of $\Delta^{(s_1)}_{exotic}$ belongs to the kernel of $\Delta^{(s_2)}_{exotic}$. We can consider the chain of L\'{e}vy Laplacians  $\Delta^{(s)}_{L}$ of order $s\in \mathbb{R}$ such that $\Delta^{(s)}_{L}=s\Delta^{(s)}_{exotic}$. Thus, the chain of the exotic Laplacians $\Delta^{(s)}_{exotic}$ can be extended for negative $s$. In this paper we consider the L\'{e}vy Laplacian 
$\Delta^{(-1)}_{L}$ in the Hida calculus. As already mentioned above, the main result of this article is as follows.  We show that the L\'{e}vy Laplacian, introduced in~\cite{Volkov2017} using the Malliavin calculus, coincides with $\pi^2\Delta^{(-1)}_{L}$ under the natural embedding of the Sobolev space into the space of generalized Hida functionals.
Unlike the classical L\'{e}vy Laplacian  $\Delta^{(1)}_{L} $ in the Hida calculus, the operator $\Delta^{(-1)}_{L}$ is little studied.

In a  particular case some of the results of this paper have been obtained in~\cite{Volkov2017a}.

The paper is organized as follows. The first section provides general information about the chain
of the L\'{e}vy Laplacians  in the deterministic case. The second section shows which elements of the chain are related to the Yang-Mills equations.  In the third section the definition of the L\'{e}vy Laplacian in the Malliavin calculus and the theorem on its connection with the Yang-Mills equations are given. In the fourth section  the definitions of the  L\'{e}vy Laplacians in the Hida calculus are given and  the theorem,  that relates the L\'{e}vy  Laplacian in the Malliavin calculus and the Laplacian $\Delta^{(-1)}_{L}$ in the Hida calculus, is proved.

\section{L\'{e}vy Laplacians}

Let $V_0$ be a locally convex space (LCS) continuously embedded in $L_2([0,1],\mathbb{R})$ in  such a way that the image of $V_0$ under the embedding is dense in $L_2([0,1],\mathbb{R})$. The symbol $\otimes_\pi$ denotes the projective tensor product of LCSs, and the symbol $\otimes$ denotes the Hilbert tensor product.
 Let $\{p_1,p_2,\ldots,p_d\}$  be an orthonormal basis in $\mathbb{R}^d$.
Let $V=\mathbb{R}^d\otimes_\pi V_0$ and $V_\mu=\{\gamma=(\gamma^\nu)_{\nu=1}^d\in V\colon \gamma^\mu\in V_0,\gamma^\nu=0 \text{ if $\nu\neq \mu$}\}$.
Then $V=V_1\oplus V_2 \oplus \ldots\oplus V_d$.
Let $M_N(\mathbb{C})$ be the space of all complex $N\times N$-matrices.  Let $(\cdot,\cdot)_{M_N(\mathbb{C})}$ be the scalar product on  $M_N(\mathbb{C})$, defined in the standard way: if  $M_1,M_2\in M_N(\mathbb{C})$, then $(M_1,M_2)_{M_N(\mathbb{C})}=tr(M_1 M_2^\ast)$.
 The symbol  $C^2(V,M_N(\mathbb{C}))$ denotes the space 
of two times Fr\'{e}chet differentiable 
$M_N(\mathbb{C})$-valued functions on $V$. For any  $x\in V$ we have $f'(x)\in M_N(\mathbb{C})\otimes_\pi V^\ast $ and $f''(x)\in M_N(\mathbb{C})\otimes_\pi L(V,V^\ast)$. (If $X_1$ and $X_2$ are complex or real LCSs, the symbol $L(X_1,X_2)$ denotes the space of all continuous linear mappings from $X_1$ to $X_2$.) The symbol
 $f''_{V_\mu V_\mu}$ denotes the second partial derivative of the function $f\in C^2(V,M_N(\mathbb{C}))$  along the space $V_\mu$.\footnote{The partial derivative $f''_{V_\mu V_\mu}(x) $ is the element from $M_N(\mathbb{C})\otimes_\pi L(V_\mu,V_\mu^\ast ) $, which we can identify with the element from $ M_N(\mathbb{C})\otimes_\pi L(V_0,V_0^\ast)$.
For a detailed exposition of the theory of differentiation on infinite-dimensional spaces see, e. g.,~\cite{BS,SmShk}.}
Let $\{e_n\}$ be an orthonormal basis in $L_2([0,1],\mathbb{R})$. We assume that $e_n\in V_0$ for each $n\in \mathbb{N}$. 

\begin{definition}
\label{def1}
The L\'evy Laplacian $\Delta^{\{e_n\},s}_{L}$  of order  $s\in\mathbb{R}$, generalized  by the basis  $\{e_n\}$,   
is a linear  mapping  from  $Dom \Delta^{\{e_n\},s}_{L}$ to the space of all $M_N(\mathbb{C})$-valued functions on  $V$ defined by:\footnote{Let $X$ be a LCS. If $F\in M_N(\mathbb{C})\otimes_\pi X^\ast$ and $f\in X$, then $<F,f>$ is an  element from $M_N(\mathbb{C})$ such that
$<F,M\otimes f>=(<F,f>,\overline{M})_{M_N(\mathbb{C})}$ for all $M\in M_N(\mathbb{C})$.}
\begin{multline}
\label{LL}
\Delta^{\{e_n\},s}_{L}f(x)=\lim_{n\to\infty}\frac 1n \sum_{k=1}^n\sum_{\mu=1}^d k^{1-s}<f''(x)p_\mu e_k,p_\mu e_k>=\\=\lim_{n\to\infty}\frac 1n \sum_{k=1}^n\sum_{\mu=1}^d k^{1-s}<f''_{V_\mu V_\mu}(x)e_k,e_k>,
\end{multline}
where $Dom  \Delta^{\{e_n\},s}_{L}$ is the space of all functions  $f\in C^2(V,M_N(\mathbb{C}))$, for which the right side of~(\ref{LL}) exists for all $x\in V$. 
\end{definition}

This definition is motivated by the following definition of the exotic L\'{e}vy Laplacians.

\begin{definition}
The exotic L\'evy Laplacian  $\Delta^{\{e_n\},s}_{exotic}$ of order $s\geq 0$, 
generalized by the basis $\{e_n\}$, is a linear mapping from   $Dom \Delta^{\{e_n\},s}_{exotic}$    
to the space of all  $M_N(\mathbb{C})$-valued functions on  $V$, defined by:
\begin{equation}
\label{exotic}
\Delta^{\{e_n\},s}_{exotic}f(x)=\lim_{n\to \infty} \frac 1{n^s} \sum_{\mu=1}^d \sum_{k=1}^n<f''(x)p_\mu e_k,p_\mu e_k>,
\end{equation}
where $Dom \Delta^{\{e_n\},s}_{exotic}$ is the space of all functions  $f\in C^2(V,M_N(\mathbb{C}))$, for which the right side of~(\ref{exotic})  exists for all  $x\in V$. 
\end{definition}
Then $\Delta^{\{e_n\},0}_{exotic}$ is the Gross-Volterra Laplacian  and $\Delta^{\{e_n\},1}_{exotic}$ is the classical L\'{e}vy Laplacian. 
\begin{proposition}
If $s>0$, then $\Delta^{\{e_n\},s}_{L}=s\Delta^{\{e_n\},s}_{exotic}$. 
\end{proposition}
\begin{proof}
The proof follows directly from the following fact (see~\cite{AJS2013,Volkov2013}).
Let $(a_n)\in \mathbb{R}^{\infty}$ and  $s>0$.  Then
\begin{equation}
\label{seq}
\lim_{n\to \infty}\frac 1 n\sum_{k=1}^na_kk^{-s+1}=s\lim_{n\to \infty}\frac 1{n^s}\sum_{k=1}^n a_k
\end{equation}
in the sense that if  one side of  equality~(\ref{seq}) exists, then the other exists  and  equality~(\ref{seq}) holds.
\end{proof}

Thus, formula~(\ref{LL}) allows us to extend the chain of the exotic L\'{e}vy Laplacians for $s\leq 0$.
The operator 
$\Delta^{\{e_n\},0}_{L}$ is different from the Gross-Volterra operator.

\begin{remark}
For the first time formula~(\ref{seq}) was used  for study of the exotic L\'{e}vy Laplacians in~\cite{AS2009}. The L\'{e}vy Laplacians $\Delta^{\{e_n\},s}_{L}$ are the particular case  of nonclassical L\'{e}vy Laplacians (see~\cite{AS2007,Volkov2013}).
\end{remark}

The following definition belongs to P.~L\'{e}vy (see~\cite{L1951,F2005,KOS}).

\begin{definition}
 An orthonormal basis  $\{e_n\}$ in  $L_2([0,1],\mathbb{R})$ is weakly uniformly dense (or equally uniformly dense), if $$\lim_{n\to\infty}\int_0^1 h(t)(\frac 1n\sum_{k=1}^n
e_k^2(t)-1)dt=0$$
for any $h\in L_{\infty}([0,1],\mathbb{R})$.
\end{definition}
Let  $h_n(t)=\sqrt{2}\sin(\pi nt)$ and  $l_n(t)=\sqrt{2}\cos(\pi nt)$ for $n\in \mathbb{N}$ and $l_0(t)=1$.
The orthonormal bases $\{h_n\}_{n=1}^\infty$ and $\{l_n\}_{n=0}^\infty$ in $L_2([0,1],\mathbb{R})$ are  weakly uniformly dense.

Let  $V_{0,0}$ be a subspace of  $V_0$ and let $V_{0,\mu}=\{\gamma=(\gamma^\nu)_{\nu=1}^d\in V_\mu\colon \gamma^\mu\in V_{0,0}\}$.
\begin{proposition}
\label{prop2}
Let $f\in C^2(V,M_N(\mathbb{C}))$. Let for any  $\mu\in\{1,\ldots,d\}$  and for all $u,v\in V_{0,\mu}$ the following holds
\begin{equation}
\label{f2d}
<f''_{V_{0,\mu} V_{0,\mu}}(x) u,v>=\int_0^1K^V_{\mu\mu}(x;s,t) u(t)v(s)dtds+\int_0^1K^L_{\mu\mu}(x;t) u(t)v(t)dt,
\end{equation}
 where $f''_{V_{0,\mu} V_{0,\mu}}$ is the second partial derivative of the function $f$ along the space $V_{0,\mu}$,  $K^V_{\mu\mu}(x,\cdot,\cdot)\in L_2([0,1]\times[0,1],M_N(\mathbb{C}))$ and $K^L_{\mu \mu}(x,\cdot)\in L_{\infty}([0,1],M_N(\mathbb{C}))$ ($K^V$ is the Volterra part and $K^L$ is the L\'{e}vy part). 
If $\{e_n\}$ is  weakly uniformly dense basis in $L_2([0,1],\mathbb{R})$ and $e_n\in V_{0,0}$ for all
$n\in \mathbb{N}$, then
\begin{equation}
\label{LLf}
\Delta^{\{e_n\},1}_{L} f(x)=\sum_{\mu=1}^d\int_0^1K^L_{\mu\mu}(x,t)dt.
\end{equation}
\end{proposition}
\begin{proof}
Since  $\{e_n\}$ is an orthonormal basis in  $L_2([0,1],\mathbb{R})$, the sequence  $\{e_n\otimes e_n\}_{n=1}^{\infty}$ is an orthonormal sequence of functions in $L_2([0,1]\times [0,1],\mathbb{R})$ and
\begin{multline}
\lim_{n\to \infty}\frac 1n \sum_{k=1}^n \sum_{\mu=1}^d\int_0^1K^V_{\mu\mu}(x;s,t) e_k(t)e_k(s)dtds=\\
=\lim_{n\to \infty}\frac 1n \sum_{k=1}^n \sum_{\mu=1}^d<K^V_{\mu\mu}(x),e_k\otimes e_k>
=\lim_{n\to \infty}\sum_{\mu=1}^d <K^V_{\mu\mu}(x),e_n\otimes e_n>=0.
\end{multline}
Then formula~(\ref{LLf}) follows from the fact that the basis   $\{e_n\}$ is  weakly uniformly dense. 
\end{proof}
\begin{remark}
One of the approaches to define the L\'{e}vy Laplacian is as follows. The  value of the L\'{e}vy
Laplacian on a function is defined as an integral functional given by the special form of the second derivative of this function (see~\cite{L1951,F2005}). If the second partial derivatives of the function $f\in C^2(V, M_N(\mathbb{C}))$ have the form~(\ref{f2d}), then the value of the classical L\'{e}vy Laplacian on $f$ can be determined by formula~(\ref{LLf}). If $V_0=C^1([0,1],\mathbb{R})$  and  $V_{0,0}=\{\gamma \in C^1([0,1],\mathbb{R})\colon \gamma(0)=\gamma(1)=0\}$, then  we obtain a generalization of the definition of the L\'{e}vy Laplacian  from  paper~\cite{AGV1994} by Accardi, Gibilisco and Volovich. The L\'evy Laplacian  $\Delta^{\{e_n\},1}_{L}$, where $\{e_n\}$ is a weakly uniformly dense basis, is an extension of the L\'evy Laplacian introduced in~\cite{AGV1994} (see~\cite{VolkovLLI}).
\end{remark}

\section{L\'{e}vy Laplacians and gauge fields}

Below, the Greek indices run through $\{1,\ldots,d\}$. In the paper we use the Einstein summation convention.

Let
$$W^{1,2}_0([0,1],\mathbb{R}^d):=\{\gamma\in AC([0,1],\mathbb{R}^d) \colon \gamma(0)=0, \dot{\gamma}\in L_2((0,1),\mathbb{R}^d)\}.$$
It is a Hilbert space with the scalar product
$$
(\gamma_1,\gamma_2)_{W^{1,2}_0([0,1],\mathbb{R}^d)}=\int_0^1(\dot{\gamma}_1(t),\dot{\gamma}_2(t))_{\mathbb{R}^d}dt.
$$
Choose $V=W^{1,2}_0([0,1],\mathbb{R}^d)$.
 Consider the classical L\'evy Laplacian  $\Delta^{\{h_n\},1}_{L}$ on  $$C^2(W^{1,2}_0([0,1],\mathbb{R}^d),M_N(\mathbb{C})).$$ 
It is this Laplacian that is associated with the  gauge fields.

Let  $A(x)=A_\mu(x)dx^\mu$ be a $C^\infty$-smooth $u(N)$-valued 1-form  on $\mathbb{R}^d$. It determines the connection in the trivial vector bundle over  $\mathbb{R}^d$ with  fibre  $\mathbb{C}^N$ and structure group  $U(N)$. If    $\phi\in C^1(\mathbb{R}^d,u(N))$, its covariant derivative is defined by $\nabla_\mu\phi=\partial_\mu\phi+[A_\mu,\phi]$.
The curvature tensor is  a $u(N)$-valued  2-form    $F(x)=\sum_{\mu<\nu}F_{\mu\nu}(x)dx^\mu\wedge dx^\nu$, where $F_{\mu\nu}=\partial_\mu A_\nu-\partial_\nu A_\mu+[A_\mu,A_\nu]$.
The Yang-Mills equations  on the connection $A$ are
\begin{equation}
\label{YMequations}
\nabla^\mu F_{\mu\nu}=0.
\end{equation}

The parallel transport 
 $U_t^A(\gamma)$  along the path $\gamma\in W^{1,2}_0([0,1],\mathbb{R}^d)$, associated with a connection  $A$, is the solution to the differential equation
 \begin{equation*}
 \label{partransp}
U^{A}_t(\gamma)=I_N-\int_0^t A_{\mu}(\gamma(s))U_s^{A}(\gamma)\dot{\gamma}^\mu(s)ds,
\end{equation*}
where $I_N$ is the identity $N\times N$-matrix. 

\begin{theorem}
\label{GF}
For the parallel transport the following equality holds
\begin{equation}
\Delta^{\{h_n\},1}_LU_1^{A}(\gamma)=-U_1^{A}(\gamma)
\int_0^1U_t^{A}(\gamma)^{-1}\nabla^{\mu}F_{\mu\nu}(\gamma(t))U^{A}_t(\gamma)\dot{\gamma}^\nu(t)dt.
\end{equation}
The connection $A$ satisfies the Yang-Mills equations~(\ref{YMequations}) if and only if
 $$\Delta^{\{h_n\},1}_LU_1^{A}=0.$$
\end{theorem}
\begin{proof}
Let $$W_\mu=\{\gamma=(\gamma^\nu)_{\nu=1}^d \in W^{1,2}_0([0,1],\mathbb{R}^d)\colon \gamma^\nu=0 \text{ if $\nu\neq \mu$}\}$$ and $W_{0,\mu}=\{\gamma\in W_\mu\colon \gamma(1)=0\}$.
The proof of the theorem is based on the fact that the second derivatives of the parallel transport along the spaces $W_{0,\mu}$
have the form:
\begin{multline*}
<(U_1^{A})''_{W_{0,\mu} W_{0,\mu}}(\gamma)u,v>=\\
=U_1^{A}(\gamma)\int_0^1dt\int_0^tdsU_t^{A}(\gamma)^{-1}F_{\mu\nu}(\gamma(t))\dot{\gamma}^\nu(t)U_t^{A}(\gamma)\times\\
\times U_s^{A}(\gamma)^{-1}F_{\mu\lambda}(\gamma(s))\dot{\gamma}^\lambda(s)U_s^{A}(\gamma)(u(t)v(s)+v(t)u(s))-\\
-U_1^{A}(\gamma)
\int_0^1U_t^{A}(\gamma)^{-1}\nabla_{\mu}F_{\mu\nu}(\gamma(t))\dot{\gamma}^\nu(t)U^{A}_t(\gamma)u(t)v(t)dt
\end{multline*}
for all $u,v\in W_{0,\mu}$. 
The assertion of the theorem can be proved by analogous arguments, as in Proposition~\ref{prop2}.
See~\cite{AGV1994} and~\cite{VolkovLLI} for the detailed proof.
\end{proof}

\begin{remark}
Theorem~\ref{GF} was first proved in~\cite{AGV1994} by  Accardi, Gibilisco and  Volovich for the L\'evy Laplacian, defined as the integral functional given by the special kind of the second derivative.
\end{remark}

Let  $D$ be an 
isometric isomorphism between  the Hilbert spaces $W^{1,2}_0([0,1],\mathbb{R}^d)$ and  $L^2([0,1],\mathbb{R}^d)$, defined by differentiation
$$
Dh(t)=\dot{h}(t), h\in W^{1,2}_0([0,1],\mathbb{R}^d).
$$
Then the L\'{e}vy Laplacian  $\Delta^{\{h_n\},1}_L$ on the space $C^2(W^{1,2}_0([0,1],\mathbb{R}^d),M_N(\mathbb{C}))$ and the L\'{e}vy Laplacian $\Delta^{\{l_n\},-1}_{L}$ on the space 
$C^2(L_2([0,1],\mathbb{R}^d),M_N(\mathbb{C}))$ are connected as follows.

\begin{proposition}
If $f\in C^2(W^{1,2}_0([0,1],\mathbb{R}^d),M_N(\mathbb{C}))$, 
then
\begin{equation}
\label{eq1}
\pi^2\Delta^{\{l_n\},-1}_{L} (f\circ D^{-1}) (x)=(\Delta^{\{h_n\},1}_{L}f)(D^{-1}x)
\end{equation}
for all $x\in L^2([0,1],\mathbb{R}^d)$.
\end{proposition}
\begin{proof}
The proof follows directly from the chain rule.
\end{proof}

In the next section we define  the L\'evy Laplacian  $\Delta_L$  in the Malliavin calculus.
This Laplacian is analogue of  the L\'{e}vy-Laplace operator $\Delta^{\{h_n\},1}_L $ on the space  $C^2(W^{1,2}_0([0,1],\mathbb{R}^d),M_N(\mathbb{C}))$. 
The classical L\'{e}vy Laplacian in the Hida calculus is  the analogue  of the L\'{e}vy Laplacian  $\Delta^{\{l_n\},1}_{L}$
 on the space  $C^2(L_2([0,1],\mathbb{R}^d),M_N(\mathbb{C}))$. 
  In the last section we  determine an analogue on the space of  Hida generalized functionals  of the L\'{e}vy Laplacian $\Delta^{\{l_n\},-1}_{L}$  on the space $C^2(L_2([0,1],\mathbb{R}^d),M_N(\mathbb{C}))$. This analogue relates to  the operator $\Delta_L$ in a way similar to~(\ref{eq1}).

\section{L\'{e}vy Laplacian in Malliavin calculus}

Let  $\{b_t\}_{t\in [0,1]}$ be a standard  $d$-dimensional Brownian motion and $(\Omega,\mathcal{F},P)$ be the associated with this process probability space ($\Omega=\{\gamma\in C([0,1],\mathbb{R}^d)\colon \gamma(0)=0\}$, $\mathcal{F}$  is the $\sigma$-algebra  generated by the Brownian motion and  $P$ is the Wiener measure).
The symbols  $db$ and  $\partial b$ denote  the  It\^o differential and the Stratonovich differential  respectively.
The space   $W^{1,2}_0([0,1],\mathbb{R}^d)$ is the Cameron-Martin space (the space of differentiability) of the Wiener measure $P$.

The Sobolev space
$W^{r,p}(P,M_N(\mathbb{C}))$  is the completion of the space of all
 $C^\infty$-smooth cylindrical $M_N(\mathbb{C})$-valued functions with compact support on $\Omega$
 with respect to the Sobolev norm
$$
\|\Phi\|_{r,p}=\sum_{k=0}^r (E(\sum_{i_1\ldots i_k=1}^\infty \|\partial_{g_{i_1}}\ldots\partial_{g_{i_k}}\Phi\|_{M_N(\mathbb{C})}^2)^{p/2})^{1/p},
$$
 where  $\{g_n\}$ is an arbitrary orthonormal basis in $W^{1,2}_0([0,1],\mathbb{R}^d)$ (for various definitions of the Sobolev spaces over the Wiener measure, see, e.g.,~\cite{Bogachev}).
For $p\geq 1$  and for any $h\in W^{1,2}_0([0,1],\mathbb{R}^d)$  the operator of differentiation along the direction $\partial_h$ can be extended by continuity as a continuous linear operator from $W^{1,p}(P,M_N(\mathbb{C}))$ to $L_p(\Omega,P;M_N(\mathbb{C}))$. We denote this extension again by the symbol $\partial_h$.
The second derivative of an element from  $W^{2,p}(P,M_N(\mathbb{C}))$ is defined by analogy.

An analogue of the classical  L\'{e}vy Laplacian $\Delta^{\{h_n\},1}_{L}$ for the Sobolev space
$W^{2,2}(P,M_N(\mathbb{C}))$
 is defined as follows.
\begin{definition}
\label{deltaL}
The L\'{e}vy Laplacian    $\Delta_L$  is a linear mapping from $Dom \Delta_L$ to $L_2(\Omega,P;M_N(\mathbb{C}))$ defined by the formula
\begin{equation}
\label{LaplaceL}
\Delta_Lf(b)=\lim_{n\to\infty}\frac 1n \sum_{k=1}^n \sum_{\mu=1}^d\partial_{p_\mu h_k}\partial_{p_\mu h_k}f(b), 
\end{equation}
where the sequence in the right  side of~(\ref{LaplaceL})  convergences strongly in  $L_2(\Omega,P;M_N(\mathbb{C}))$ and 
$Dom \Delta_L$  consists of all  $f\in W^{2,2}(P,M_N(\mathbb{C}))$ for which the right side of~(\ref{LaplaceL}) exists.
\end{definition}

The stochastic parallel transport 
 $U^{A}(b,t)$, associated   with the connection $A$, is a solution to the stochastic equation
 in the  sense of Stratonovich:
\begin{equation*}
\label{partransp}
U^{A}(b,t)=I_N-\int_0^t A_{\mu}(b_s)U^{A}(b,s)\partial b^{\mu}_s.
\end{equation*}
If $A$ and its partial derivatives of the first and the second order are bounded, this equation has the unique strong solution.

\begin{theorem}
\label{thm1}
Let a connection $A$ be  bounded together with all its partial derivatives up to the third
order inclusive. 
Then for the stochastic parallel transport the following holds
\begin{multline}
\label{laplU}
\Delta_LU^{A}(b,1)=U^{A}(b,1)\int_0^1 U^{A}(b,t)^{-1}F_{\mu\nu}(b_t)F^{\mu\nu}(b_t)U^{A}(b,t)dt-
\\-U^{A}(b,1)\int_0^1U^{A}(b,t)^{-1}\nabla^{\mu}F_{\mu\nu}(b_t)U^{A}(b,t)d b^\nu_t.
$$
\end{multline}
The connection  $A$ is a solution to the Yang-Mills equations~(\ref{YMequations}) if and only if for the stochastic parallel transport the following holds
 $$\Delta_LU^{A}(b,1)=U^{A}(b,1)\int_0^1 U^{A}(b,t)^{-1}F_{\mu\nu}(b_t)F^{\mu\nu}(b_t)U^{A}(b,t)dt.$$
\end{theorem}
For the proof see~\cite{Volkov2017}.

\begin{remark}
In the paper~\cite{LV2001} by  Leandre and  Volovich, the L\'{e}vy Laplacian  has been  introduced on the Sobolev space over the Wiener measure 
on the space of paths in a compact Riemannian manifold. This Laplacian has been defined as an integral functional given by the special form of the second derivative. The value of such  a L\'{e}vy Laplacian on the stochastic parallel transport does not have the first term in the right-hand side of~(\ref{laplU}). Thus, for this L\'{e}vy Laplacian the theorem on the equivalence of the Yang-Mills equations and the Laplace equation for the L\'{e}vy Laplacian is satisfied.
It would be interesting to study the relationship between the Laplacian given by Definition~\ref{deltaL} and the Laplacian introduced by  Leandre and Volovich.
\end{remark}

\begin{remark}
In the paper~\cite{Volkov2017}, a divergence corresponding to the Laplacian $\Delta_L$ has been introduced. It has been shown that a stochastic parallel transport is a solution to an equation containing such a divergence if and only if the associated connection is a solution to the Yang-Mills equations. The resulting equation for the stochastic parallel transport is an analogue of the equation of motion of chiral fields (cf.~\cite{AV1981}).
\end{remark}

\begin{remark}
It would be interesting to investigate whether it is possible to use the Levy-Laplacian approach in some   areas connected to the theory of gauge fields (see, e.g.,~\cite{ABT,Sergeev,MarShir,Zharinov,Katanaev}).
\end{remark}

At the end of this section, we give some general information about the  Fock spaces and the Wiener-Ito-Segal isomorphism, which are needed for the proof of the theorem in the next section.

Let $\frak H$ be  a separable Hilbert space. 
  The symbol $\widehat{\otimes}$ denotes the symmetric tensor product.
If  $F\in \frak{H}^{\widehat{\otimes}n}$  and  $f\in \frak{H}^{\widehat{\otimes}k}$,  than the  contraction  $F\widehat{\otimes}_kf$ is  an  element from $\frak{H}^{\widehat{\otimes}(n-k)}$ such
that  for any  $h\in \frak{H}^{\widehat{\otimes}k}$ holds $<F,h\widehat{\otimes} f>=<F\widehat{\otimes}_k f,h>$.
The following estimates for the norms hold (see, e.g.,~\cite{Obata1994})
\begin{equation}
\label{tens1}
\|F\widehat{\otimes}f\|_{\frak{H}^{\widehat{\otimes}(n+k)}}\leq \|F\|_{\frak{H}^{\widehat{\otimes}n}}\|f\|_{\frak{H}^{\widehat{\otimes}k}},
\end{equation}
\begin{equation}
\label{tens2}
\|F\widehat{\otimes}_k f\|_{\frak{H}^{\widehat{\otimes}(n-k)}}\leq \|F\|_{\frak{H}^{\widehat{\otimes}n}}\|f\|_{\frak{H}^{\widehat{\otimes}k}},  \text{\, if $n\geq k$.}
\end{equation}
The  (boson) Fock space $\Gamma(\frak{H})$ over  $\frak{H}$ is a Hilbert space with  the Hilbert norm $\|\cdot\|_{\Gamma(\frak{H})}$, defined by 
$$
\Gamma(\frak{H})=\{f=(f_n)_{n=0}^\infty;\; f_n\in
\frak{H}^{\widehat{\otimes}n}, \|f\|_{\Gamma(\frak{H})}=\sum_{n=0}^\infty
n!|f_n|_{\frak{H}^{\widehat{\otimes}n}}^2<\infty\}.
$$  
 The tensor product 
 $M_N(\mathbb{C})\otimes\Gamma(\frak{H})$ and the associated Hilbert norm  
are of the form
\begin{multline*} 
M_N(\mathbb{C})\otimes\Gamma(\frak{H})=\\=\{F=(F^n)_{n=0}^\infty;\; F^n\in M_N(\mathbb{C})\otimes
\frak{H}^{\widehat{\otimes}n}, \|F\|_{M_N(\mathbb{C})\otimes\Gamma(\frak{H})}=\sum_{n=0}^\infty
n!|F^n|_{M_N(\mathbb{C})\otimes\frak{H}^{\widehat{\otimes}n}}^2<\infty\}.
\end{multline*}

Let  $T_n=T_n(\mathbb{R}^d,M_N(\mathbb{C}))$ be the space of all $M_N(\mathbb{C})$-valued tensors of type  $(0,n)$ on
$\mathbb{R}^d$, endowed with the standard structure of a Hilbert space.\footnote{If $F^n\in T_n$, then $\|F^n\|^2_{T_n}=\sum_{i_1=1}^d \ldots \sum_{i_n=1}^d tr(F^n_{i_1i_2\ldots i_n}(F^n_{i_1i_2\ldots i_n})^\ast)$} Let  the space   $L^{sym}_2([0,1]^n,T_n)$ consists of all functions  $F^n\in L_2([0,1]^n,T_n)$, for which the following holds 
$$
F^n_{i_{\sigma (1)}i_{\sigma (2)}\cdots i_{\sigma (k)}}(t_{\sigma (1)},t_{\sigma (2)},\ldots, t_{\sigma (n)} )=F^n_{i_1i_2\cdots i_k}(t_1,t_2,\ldots, t_n),
$$
where $\sigma$ is a permutation of the first $n$ natural numbers.
If $\frak{H}=L_2([0,1],\mathbb{R}^d)$, 
then  $M_N(\mathbb{C})\otimes \frak{H}^{\widehat{\otimes}n}$ coincides with the space  $L^{sym}_2([0,1]^n,T_n)$.

The  unitary Wiener-Ito-Segal isomorphism  $\mathcal{J}_1$ between $M_N(\mathbb{C})\otimes \Gamma(H) $ and $L_2(\Omega,P;M_N(\mathbb{C}))$ acts as follows. If $F=(F^n)_{n=0}^\infty\in  M_N(\mathbb{C}) \otimes \Gamma(H)$, then
\begin{equation}
\label{WIS}
\mathcal{J}_1(F)=\sum_{n=0}^\infty I_n(F^n),
\end{equation}
where $I_0(F^0)=F^0$,
\begin{equation*}
I_n(F^n)
=\sum_{i_1=1}^d \ldots \sum_{i_n=1}^d n!\int_0^1\int_0^{t_n}\ldots \int_0^{t_2}F^n_{i_1\ldots i_n}(t_1,\ldots t_n)db^{i_1}_{t_1}\ldots db^{i_n}_{t_n}
\end{equation*}
and the series~(\ref{WIS}) converges strongly in $L_2(\Omega,P;M_N(\mathbb{C}))$.
It is known (see, e.g.,~\cite{Nualart,Bogachev}) that  if  $\mathcal{J}_1(F)\in W^{1,2}(P,M_N(\mathbb{C}))$ and $h\in W^{1,2}_0([0,1],\mathbb{R}^d)$, then
\begin{equation}
\label{derivative}
\partial_h(\mathcal{J}_1(F))=\sum_{n=1}^\infty I_{n-1}(F^n\widehat{\otimes}_1 \dot{h}).
\end{equation}

\section{L\'{e}vy Laplacians in Hida calculus}

 We need to generalize Definition~\ref{def1} from Section 1 to determine the L\'{e}vy Laplacian on the space of generalized Hida functionals.

Let  $T$ be an interval in $\mathbb{R}$ and  $[0,1]\subset T$. 
Let   $\mathcal V_0$  be a complex LCS continuously embedded in  $L_2(T,\mathbb{C})$  in  such a way that the image  of $\mathcal V_0$  under the embedding is dense in $L_2(T,\mathbb{C})$. Let $\mathcal V=\mathbb{C}^d \otimes_\pi \mathcal V_0$
and $H_\mathbb{C}=L_2(T,\mathbb{C}^d)$.
Let  $C^2_{L}(\mathcal V,M_N(\mathbb{C}))$ be the space of all two times  Fr\'{e}chet complex differentiable  $M_N(\mathbb{C})$-valued functions on $\mathcal V$ (for  the  definition of the  Fr\'{e}chet complex differentiability see, e.g.,~\cite{SmShk}), the second derivative of which has the form
\begin{equation*}
\label{Lfunctional}
<f''(x)u,v>=\int_{T}\int_{T} K^V_{\mu\nu}(x;s,t)u^\mu(t)v^\nu(s)dtds
+\int_{T}K^L_{\mu\nu}(x;t)u^\mu(t)v^\nu(t)dt
\end{equation*}
for all $u,v \in \mathcal V$, where  $K^V(x;\cdot,\cdot)\in L_2^{sym}(T\times T,T_2))$, $K^L_{\mu \nu}(x;\cdot)\in L_{\infty}(T,M_N(\mathbb{C}))$ and $K^L_{\mu\nu}=K^L_{\nu \mu}$ ($K^V$ is the Volterra part and $K^L$ is the L\'{e}vy part as in Section 1). 
Then for each $x\in \mathcal V$ it is possible to extend by continuity $f''(x)$  to the  element  from $M_N(\mathbb{C})\otimes L(H_\mathbb{C}, H_\mathbb{C})$, which we will denote by the same symbol $f''(x)$.
Let  $\{e_n\}$ be an orthonormal basis in  $L_2([0,1],\mathbb{R})$, 
whose elements may no  belong to the space $\mathcal V_0$. (We identify
$e_n$ with the function from $ L_2(T,\mathbb{R})$ that is equal to $e_n(t) $ for $t\in[0,1]$ and is equal to zero for $t\in T\setminus [0,1]$.)  For $s=\{-1,1\}$ the definition of the L\'{e}vy Laplacian $\Delta^{\{e_n\},s}_{L}$  can be generalized as follows.
\begin{definition}
 The L\'{e}vy Laplacian  $\Delta^{\{e_n\},s}_{L}$ of order  $s\in \{-1,1\}$ is a linear mapping from  $Dom \Delta^{\{e_n\},s}_{L}$ to the space of all $M_N(\mathbb{C})$-valued functions on $\mathcal V$ defined by:
\begin{equation}
\label{LLL1}
\Delta^{\{e_n\},s}_{L}f(x)=\lim_{n\to\infty}\frac 1n \sum_{k=1}^n\sum_{\mu=1}^d k^{1-s}<f''(x)p_\mu e_k,p_\mu e_k>,
\end{equation}
where  $Dom  \Delta^{\{e_n\},s}_{L}$ is the space of all functions   $f\in C^2_{L}(\mathcal V,M_N(\mathbb{C}))$, for which the right side of~(\ref{LLL1}) exists for any $x\in \mathcal V$. 
\end{definition}

Let $\{\textbf{e}_n\}$ be an orthonormal basis in $L_2(T,\mathbb{R})$.
 Let  $\textbf{A}$ be a self-adjoint operator on $H_\mathbb{C}$ acting  by the formula:
$$
\textbf{A}(\textbf{e}_n\otimes p_\mu)=\lambda_n (\textbf{e}_n\otimes p_\mu),
$$
where $\{\lambda_n\}$ 
is an increasing sequence of real numbers such that
$$
1<\lambda_1\leq\lambda_2\leq\ldots\leq\lambda_n\leq \ldots\;\text{and}\;\sum_{k=1}^\infty\lambda_k^{-2}<\infty.
$$
Then the operator   $\textbf{A}^{-1}$  is a Hilbert-Schmidt operator.
For any  $p\geq 0$ the Hilbert space $E_p$ with the  Hilbert norm $|\cdot|_p$ is defined by
$$
E_p=\{\xi \in H_\mathbb{C}\colon |\xi|^2_p=\sum_{k=1}^\infty \sum_{\mu=1}^d\lambda_k^{2p}|(\xi,\textbf{e}_k\otimes p_\mu)_{H_\mathbb{C}}|^2< \infty\}.
$$
For any  $p<0$ the norm $|\cdot|_p$  is defined on all $H_\mathbb{C}$  by 
$$
|\xi|^2_p=\sum_{k=1}^\infty \sum_{\mu=1}^d\lambda_k^{2p}|(\xi,\textbf{e}_k\otimes p_\mu)_{H_\mathbb{C}}|^2.
$$
For   $p<0$ the Hilbert space  $E_p$ is the completion of $H_\mathbb{C}$ with respect to this norm.
For all $p\in \mathbb{R}$  we denote  the Hilbert norm on  $E_p^{\widehat{\otimes}n}$ also by the symbol $|\cdot|_p$. The space $E_0$ coincides with $H_\mathbb{C}$. Below we denote the Hilbert norm on $M_N(\mathbb{C})\otimes H_\mathbb{C}^{\widehat{\otimes}n}$    also by the symbol $|\cdot |_0$.
This will not lead to any  confusion.

Denote the projective limit  $\projlim_{p\to +\infty}E_p$ by  $E_\mathbb{C}$. The space $E_\mathbb{C}$ is a nuclear Fr\'{e}chet space  and, hence, is a reflexive space. Its conjugate space
$E_\mathbb{C}^\ast$ is the inductive limit $\indlim_{p\to +\infty}E_{-p}$.  
We obtain the complex rigged  Hilbert space:
$$
E_\mathbb{C}\subset H_\mathbb{C}\subset E^\ast_\mathbb{C}.
$$

Using the restriction of the operator $\textbf{A}$ to $H_\mathbb{R}=L_2(T,\mathbb{R}^d)$ 
in a similar way we obtain the real rigged Hilbert space:
$$
E_\mathbb{R}\subset H_\mathbb{R}\subset E^\ast_\mathbb{R}.
$$

The norms on the spaces $\Gamma(E_p)$ and $M_N(\mathbb{C})\otimes \Gamma(E_p)$ are denoted by the symbol $\|\cdot\|_p$.
Denote the projective limit    $\projlim_{p\to +\infty}\Gamma(E_p)$ by the symbol  $\mathcal E$.
Then $\mathcal E^\ast=\indlim_{p\to +\infty}
\Gamma(E_{-p})$.
The space  $M_N(\mathbb{C})\otimes_\pi \mathcal E=\projlim_{p\to +\infty}M_N(\mathbb{C})\otimes\Gamma(E_p)$  is the space of  $M_N(\mathbb{C})$-valued Hida  test functionals (white noise test functionals) and the space  $(M_N(\mathbb{C})\otimes_\pi \mathcal E)^\ast$ is the space of  $M_N(\mathbb{C})$-valued Hida generalized functionals (white noise generalized functionals).
Denote the canonical bilinear form on $(M_N(\mathbb{C})\otimes_\pi\mathcal E)^\ast\times (M_N(\mathbb{C})\otimes_\pi \mathcal E)$ by the symbol $\langle\langle \cdot,\cdot\rangle\rangle$.

\begin{remark}
If $T=\mathbb{R}$ and
$$
\textbf{A}=1+t^2-\frac {d^2}{dt^2},
$$
then $E_\mathbb{C}=S(\mathbb{R},\mathbb{C}^d)$  is the Schwartz space of rapidly decreasing functions  and
$E^\ast_\mathbb{C}=S^\ast(\mathbb{R},\mathbb{C}^d)$ is the space of generalized functions of slow growth. The case, where $T=[0,1]$ and $\{\textbf{e}_n\}$ is a basis consisting of trigonometric
functions, is also often considered  (see~\cite{AJS2009,AJS2011,AJS2013,Volkov2013}).
\end{remark}

By the  Minlos-Sazonov theorem there is a Gaussian probability measure $\mu_I$ on $\sigma$-algebra generated by $E_\mathbb{R}$-cylindrical sets on $E_\mathbb{R}^\ast$ such that its Fourier transform has the form
$\widetilde{\mu_I}(\xi)=e^{-\frac{(\xi,\xi)_{H_\mathbb{R}}}{2}}$.
An element from $\mathcal E$
$$
\psi_\xi=(1,\xi,\frac{\xi^{\otimes 2}}{2},\ldots,\frac{\xi^{\otimes
n}}{n!},\ldots), 
$$
where $\xi\in E_\mathbb{C}$, is called a coherent state. 

The unitary Wiener-Ito-Segal isomorphism $j_2$ between  $\Gamma(H_\mathbb{C})$ and
$L_2(E_\mathbb{R}^\ast,\mu_I;\mathbb{C})$ is uniquely determined by the values on coherent states
$$
j_2(\psi_\xi)(x)=e^{\langle
x,\xi\rangle-{\langle\xi,\xi\rangle}/{2}}.
$$
The complex rigged Hilbert space
$$
\mathcal E\subset \Gamma(H_\mathbb{C})\cong L_2(E_\mathbb{R}^\ast,\mu_I,\mathbb{C})\subset \mathcal E^\ast
$$ 
is called  Hida-Kubo-Takenaka space.
 The symbol $\mathcal{J}_2$ denotes the  unitary isomorphism  between $M_N(\mathbb{C})\otimes\Gamma(H_\mathbb{C})$ and 
 $L_2(E_\mathbb{R}^\ast,\mu_I;M_N(\mathbb{C}))$, generated  by $j_2$.
We will not distinguish  $F\in M_N(\mathbb{C})\otimes\Gamma(H_\mathbb{C})$
and $J_2(F)\in L_2(E_\mathbb{R}^\ast,\mu_I;M_N(\mathbb{C}))$ and respectively  the spaces
$M_N(\mathbb{C})\otimes\Gamma(H_\mathbb{C})$ and $L_2(E_\mathbb{R}^\ast,\mu_I;M_N(\mathbb{C}))$. 
 
\begin{remark}
Let  $\nu_I$ be the Gaussian measure with  zero mean and  identity correlation operator
on  $L_2(T,\mathbb{R}^d)$.  This measure is not $\sigma$-additive. 
The measure $\mu_I$ is the image of a measure 
$\nu_I$ under the embedding of  $L_2(T,\mathbb{R}^d)$ into $E_\mathbb{R}^\ast$. Let  $T=[0,1]$.  Let  the measure $\nu_I\circ D$ be the image of the measure $\nu_I$ under $D^{-1}$. Then the Wiener measure $P$ is the image of the measure   $\nu_I\circ D$  under the embedding of   $W^{1,2}_0([0,1],\mathbb{R}^d)$ into $\Omega$.
\end{remark} 
 
The Hida generalized functional  $\Phi\in\mathcal E^\ast$ can be formally written in the form (see~\cite{Obata1994,K})
\begin{equation}
\label{WSI}
\Phi=\sum_{n=0}^\infty<:x^{\otimes n}:,F^n>,
\end{equation}
where $:x^{\otimes n}:$ is the Wick tensor of order $n$, and   $F^n\in M_N\otimes_\pi(E_\mathbb{C}^{\widehat{\otimes}_\pi n})^\ast$ (the space $(E_\mathbb{C}^{\widehat{\otimes}_\pi n})^\ast$ coincides with $\indlim_{p\to +\infty}
E_{-p}^{\widehat{\otimes} n}$).

The $S$-transform of a generalized functional   $\Phi\in M_N(\mathbb{C})\otimes_\pi\mathcal E^\ast$  is the function $S\Phi\colon E_\mathbb{C} \to M_N(\mathbb{C})$ defined by  $$(S\Phi(\xi),\overline{M})_{M_N(\mathbb{C})}=\langle\langle\Phi,M\otimes\psi_{\xi}\rangle\rangle,$$ for all $\xi\in
E_\mathbb{C}$ and $M\in M_N(\mathbb{C})$. The generalized Hida functional is uniquely determined by its $S$-transform.
If $\Phi\in M_N(\mathbb{C})\otimes_\pi\mathcal E^\ast$ has the form~(\ref{WSI}),
the following holds
\begin{equation}
\label{Stranf}
S\Phi(\xi)=\sum_{n=0}^\infty<F^n,\xi^{\otimes n}>.
\end{equation}
A $M_N(\mathbb{C})$-valued function  $G$ on $E_\mathbb{C}$ is the $S$-transform of some $\Phi\in \mathcal E^\ast$ if and only if  (see, e.g.,~\cite{K,Obata1994})
\begin{enumerate}
\item for any  $\zeta,\eta \in E_\mathbb{C}$ the function  $G_{\zeta,\eta}(z)=G(z\eta+\zeta)$ is entire;
\item there exist  $C,K>0$ and  $p\in \mathbb{R}$ such that for all $\xi\in E_\mathbb{C}$ 
the following estimate holds
$$
\|G(\xi)\|_{M_N(\mathbb{C})}\leq C e^{K|\xi|_p^2}.
$$
\end{enumerate}
A $M_N(\mathbb{C})$-valued function on $E_\mathbb{C}$, that  satisfies conditions 1 and 2,
is called  $U$-functional. The symbol $\mathcal{F}_U$ denotes the space of  
$U$-functionals.

\begin{remark}
For the detailed theory of the vector-valued Hida distributions see~\cite{Obata1994}. 
\end{remark}

\begin{definition}
The domain of the L\'{e}vy Laplacian  $\widetilde{\Delta}_L^{\{e_n\},s}$ of order $s\in\{-1,1\}$,
generalized by the orthonormal basis $\{e_n\}$ in $L_2([0,1],\mathbb{R})$, is the space $Dom \widetilde{\Delta}_L^{\{e_n\},s}=\{\Phi\in(M_N(\mathbb{C})\otimes_\pi \mathcal E)^\ast \colon S\Phi\in Dom  \Delta^{\{e_n\},s}_{L}, \Delta^{\{e_n\},s}_L S\Phi \in \mathcal{F}_U\}$. The L\'{e}vy Laplacian $\widetilde{\Delta}_L^{\{e_n\},s}$  is a linear mapping from  $Dom \widetilde{\Delta}_L^{\{e_n\},s}$ to $M_N(\mathbb{C}) \otimes_\pi  \mathcal E^\ast$ defined by
\begin{equation}
\label{LLPhi}
\widetilde{\Delta}^{\{e_n\},s}_L\Phi=S^{-1} \Delta^{\{e_n\},s}_{L} (S\Phi).
\end{equation}
\end{definition}

\begin{remark}
In the present paper we consider only L\'evy Laplacians of orders $(-1)$ and $1$.  The chain of the exotic L\'evy Laplacians generated by a basis of trigonometric functions in the Hida calculus has been considered in papers~\cite{AJS2009,AJS2011,AJS2013}. Its extension for negative orders has been considered in the paper~\cite{Volkov2013}. In fact, the L\'evy Laplacian of order $(-1)$ was first considered in~\cite{AS2007}.
\end{remark}

The following fact is known. The proof, which we give,  is taken from Theorem 6.42 from~\cite{HKPS} 
with minor modifications.

\begin{proposition}
\label{prop1}
If $\Phi\in L_2(E^\ast_{\mathbb{R}},\mu_I;M_N(\mathbb{C}))$,then $S\Phi \in C^2_{L}(E_\mathbb{C},M_N(\mathbb{C}))$ and the L\'{e}vy part of  $S\Phi''$ vanishes, i.e.
$K^L_{\mu\nu}=0$. Hence, if $\{e_n\}$ is  weakly uniformly dense basis, then $\widetilde{\Delta}^{\{e_n\},1}_L\Phi=0$.
\end{proposition}
\begin{proof}
Let $\Phi \in L_2(E^\ast_{\mathbb{R}},\mu_I;M_N(\mathbb{C}))$ and $\Phi=\sum_{n=0}^\infty<:x^{\otimes n}:,F^n>$.
 By the nuclear theorem (see, e.g.,~\cite{Obata1994,BS})  $L(E_\mathbb{C},E^\ast_\mathbb{C})\cong (E_\mathbb{C}\otimes_\pi E_\mathbb{C})^\ast$. Thus one can identify $S\Phi''(\xi)$ with the element from $M_N(\mathbb{C})\otimes_\pi (E_\mathbb{C}\widehat{\otimes}_\pi E_\mathbb{C})^\ast$. 
By~(\ref{Stranf}) we have
\begin{multline*}
<S\Phi''(\xi) ,\zeta \otimes \eta>=\sum_{n=2}^\infty n(n-1)<F^{n},\xi^{\otimes (n-2)}\widehat{\otimes} \zeta\widehat{\otimes} \eta>=\\
=\sum_{n=2}^\infty n(n-1)<F^{n}\widehat{\otimes}_{(n-2)}\xi^{\otimes (n-2)},\zeta\widehat{\otimes} \eta>.
\end{multline*}
We consider the entire function 
$$
q(z)=\sum_{n=2}^\infty  \frac{n^2(n-1)^2}{n!}z^{n-2}.
$$
By the Schwarz  inequality and  estimates~(\ref{tens1}) and~(\ref{tens2}) we have that
\begin{multline*}
\sum_{n=2}^\infty |n(n-1)F^{n}\widehat{\otimes}_{(n-2)}\xi^{\otimes (n-2)}|_{0}
\leq \sum_{n=2}^\infty {\sqrt{n!}} |F^{n}|_{0} \frac{n(n-1)}{\sqrt{n!}}|\xi|^{(n-2)}_{0}\leq\\
\\ \leq (\sum_{n=2}^\infty n!|F^{n}|^2_{0})^{\frac 12}(\sum_{n=2}^\infty  \frac{n^2(n-1)^2}{n!}|\xi|^{2(n-2)}_{0})^{\frac 12} \leq\|F\|_{0} \sqrt{q(|\xi|_0^2)}.
\end{multline*}
We obtain that the series 
$$
\sum_{n=2}^\infty n(n-1)F^{n}\widehat{\otimes}_{(n-2)}\xi^{\otimes (n-2)}
$$
converges in $M_N(\mathbb{C})\otimes H_{\mathbb{C}}^{\widehat{\otimes} 2}$ for all $\xi\in E_\mathbb{C}$.
Then $S\Phi''(\xi)\in M_N(\mathbb{C})\otimes H_{\mathbb{C}}^{\widehat{\otimes} 2}$ and there exists $K^V(\xi;\cdot,\cdot)\in L^{sym}_2(T\times T,T_2)$ such that
$$
<S\Phi''(\xi)\zeta,\eta>=<S\Phi''(\xi),\zeta\widehat{\otimes}\eta>=\int_{T}\int_{T} K^V_{\mu\nu}(\xi;s,t)\zeta^\mu(t)\eta^\nu(s)dtds.
$$
This means that  $S\Phi \in C^2_{L}(E_\mathbb{C},M_N(\mathbb{C}))$ and the  L\'{e}vy part of the second derivative of $S\Phi$ vanishes.
\end{proof}

Denote  the  embedding of  $M_N(\mathbb{C})\otimes\Gamma(L_2([0,1],\mathbb{R}^d))$ into $M_N(\mathbb{C})\otimes\Gamma(L_2(\mathbb{R},\mathbb{R}^d))$ by the symbol  $\mathcal{J}_3$  and  the orthogonal projection $M_N(\mathbb{C})\otimes\Gamma(L_2(\mathbb{R},\mathbb{R}^d))$ on $M_N(\mathbb{C})\otimes\Gamma(L_2([0,1],\mathbb{R}^d))$ by the symbol $\mathcal P$.  Let the linear mapping  $\mathcal{J}\colon L_2(\Omega,P;M_N(\mathbb{C}))\to L_2(E^\ast_{\mathbb{R}},\mu_I;M_N(\mathbb{C}))$ be defined by
$$
\mathcal{J}=\mathcal {J}_3 \mathcal {J}_1^{-1}.
$$

\begin{theorem}
If $\Psi\in Dom\Delta_L$, then
\begin{equation}
\label{main}
\mathcal J\Delta_L \Psi=\pi^2\widetilde{\Delta}^{\{l_n\},-1}_{L}\mathcal J\Psi.
\end{equation}
\end{theorem}
\begin{proof}
Let $\Psi=\sum_{n=0}^\infty I_n(F^n)\in Dom\Delta_L$. 
For all $\xi\in E_\mathbb{C}$ and $M\in M_N(\mathbb{C})$ we have
\begin{multline}
\label{final22}
\lim_{n\to\infty}\frac 1n \sum_{k=1}^n\sum_{\mu=1}^d (S(\mathcal{J}\partial_{p_\mu h_k}\partial_{p_\mu h_k}\Psi)(\xi),\overline{M})_{M_N(\mathbb{C})}=\\=\lim_{n\to\infty}<<\mathcal{J}(\frac 1n \sum_{k=1}^n \sum_{\mu=1}^d\partial_{p_\mu h_k}\partial_{p_\mu h_k}\Psi),M\otimes\psi_\xi>>=\\=
\lim_{n\to\infty}(\frac 1n \sum_{k=1}^n \sum_{\mu=1}^d\partial_{p_\mu h_k}\partial_{p_\mu h_k}\Psi,\overline{\mathcal P (M\otimes \psi_{\xi})})_{L_2(\Omega,P;M_N(\mathbb{C}))}=\\=(\Delta_L \Psi,\overline{\mathcal P (M\otimes \psi_\xi)})_{L_2(\Omega,P;M_N(\mathbb{C}))}
=<<\mathcal J(\Delta_L \Psi),M\otimes\psi_\xi>>=\\=(S(\mathcal J\Delta_L \Psi)(\xi),\overline{M})_{M_N(\mathbb{C})}.
\end{multline} 
Since  $\dot{h}_k(t)=\pi k l_k(t)$,  equality~(\ref{derivative}) implies 
$$
\partial_{p_\mu h_k}\partial_{p_\mu h_k}\Psi=\pi^2 k^2 \sum_{n=2}^\infty n(n-1) I_{n-2}
(F^{n}\widehat{\otimes}_2(p_\mu l_k\otimes p_\mu l_k)).
$$
Then
$$
\mathcal {J}\partial_{p_\mu h_k}\partial_{p_\mu h_k}\Psi=\pi^2 k^2 \sum_{n=2}^\infty n(n-1) <:x^{\otimes n-2}:,
F^{n}\widehat{\otimes}_2(p_\mu l_k\otimes p_\mu l_k)>.
$$
Due to~(\ref{Stranf}) we obtain
$$
S(\mathcal {J}\partial_{p_\mu h_k}\partial_{p_\mu h_k}\Psi)(\xi)=\pi^2 k^2 \sum_{n=2}^\infty n(n-1) <F^{n}\widehat{\otimes}_2(p_\mu l_k\otimes p_\mu l_k),\xi^{\otimes n-2}>.
$$ 
Since $\mathcal J\Psi\in L_2(E^\ast_{\mathbb{R}},\mu_I;M_N(\mathbb{C}))$, Proposition~\ref{prop1}
implies
 $$S(\mathcal J\Psi)''(\xi)=\sum_{n=2}^\infty n(n-1)F^{n} \widehat{\otimes}_{(n-2)} \xi^{\otimes n-2}\in  M_N(\mathbb{C})\otimes H^{\widehat{\otimes}2}_{\mathbb{C}}.$$
Then
\begin{multline*}
\pi^2 k^2 <S(\mathcal J\Psi)''(\xi),p_\mu l_k \otimes p_\mu l_k>=\sum_{n=2}^\infty \pi^2 k^2 n(n-1)<F^{n} \widehat{\otimes}_{(n-2)} \xi^{\otimes n-2},p_\mu l_k \otimes p_\mu l_k>=\\
=\sum_{n=2}^\infty \pi^2 k^2 n(n-1) <F^{n}\widehat{\otimes}_2(p_\mu l_k\otimes p_\mu l_k),\xi^{\otimes n-2}>=S(\mathcal {J}\partial_{p_\mu h_k}\partial_{p_\mu h_k}\Psi)(\xi).
\end{multline*}
Thus
\begin{multline}
\label{final1}
\lim_{n\to\infty}\frac 1n \sum_{k=1}^n\sum_{\mu=1}^d S(\mathcal {J}\partial_{p_\mu h_k}\partial_{p_\mu h_k}\Psi)(\xi)=\\=\lim_{n\to\infty}\frac 1n \sum_{k=1}^n \sum_{\mu=1}^d\pi^2 k^2 <S(\mathcal J\Psi)''(\xi)p_\mu l_k,p_\mu l_k>=\pi^2\Delta_L^{\{l_n\},-1}S(\mathcal J\Psi)(\xi).
\end{multline}
From~(\ref{final1}) and~(\ref{final22})  we obtain that
$$
S(\mathcal J\Delta_L \Psi)(\xi)=\pi^2\Delta_L^{\{l_n\},-1}S(\mathcal J\Psi)(\xi).
$$
This means that  equality~(\ref{main}) is true.
\end{proof}

\section*{Acknowledgments}

The author thanks  O.G. Smolyanov and  I.V. Volovich  for useful discussions.

This work is supported by the Russian Science Foundation under grant 14-50-00005.

\end{document}